\documentclass[10pt,conference]{IEEEtran}

\usepackage{amsfonts}
\usepackage{amsmath}
\usepackage{amssymb}
\usepackage{lscape}
\usepackage{epsf}
\usepackage{graphicx}

\newtheorem{theorem}{\indent Theorem}[section]

\newtheorem{proposition}[theorem]{\indent Proposition}

\newtheorem{EXAMPLE}{\indent Example}[section]
\newtheorem{definition}{\indent Definition}[section]

\newcommand{\code}{{\mathcal{C}}}

\newcommand{\cL}{{\mathcal{L}}}
\newcommand{\cV}{{\mathcal{V}}}
\newcommand{\cG}{{\mathcal{G}}}
\newcommand{\cH}{{\mathcal{H}}}

\newcommand{\cA}{{\mathcal{A}}}
\newcommand{\cB}{{\mathcal{B}}}
\newcommand{\cD}{{\mathcal{D}}}
\newcommand{\cI}{{\mathcal{I}}}
\newcommand{\cJ}{{\mathcal{J}}}
\newcommand{\cM}{{\mathcal{M}}}

\newcommand{\cK}{{\mathcal{K}}}
\newcommand{\cP}{{\mathcal{P}}}
\newcommand{\cQ}{{\mathcal{Q}}}
\newcommand{\cR}{{\mathcal{R}}}
\newcommand{\cS}{{\mathcal{S}}}

\newcommand{\cU}{{\mathcal{U}}}

\newcommand{\cY}{{\mathcal{Y}}}
\newcommand{\cYY}{{\mathcal{Y}}}

\newcommand{\ip}{{\mathrm{ip}}}
\newcommand{\op}{{\mathrm{op}}}

\newcommand{\bldb}{{\mbox{\boldmath $b$}}}
\newcommand{\bldbb}{{\mbox{\scriptsize \boldmath $b$}}}
\newcommand{\bldc}{{\mbox{\boldmath $c$}}}

\newcommand{\bldd}{{\mbox{\boldmath $d$}}}
\newcommand{\blddd}{{\mbox{\scriptsize \boldmath $d$}}}

\newcommand{\bldg}{{\mbox{\boldmath $g$}}}
\newcommand{\bldgg}{{\mbox{\scriptsize \boldmath $g$}}}
\newcommand{\bldh}{{\mbox{\boldmath $h$}}}

\newcommand{\bldp}{{\mbox{\boldmath $p$}}}

\newcommand{\bldP}{{\mbox{\boldmath $P$}}}

\newcommand{\bldv}{{\mbox{\boldmath $v$}}}
\newcommand{\bldV}{{\mbox{\boldmath $V$}}}

\newcommand{\bldW}{{\mbox{\boldmath $W$}}}

\newcommand{\bldx}{{\mbox{\boldmath $x$}}}
\newcommand{\bldxx}{{\mbox{\scriptsize \boldmath $x$}}}
\newcommand{\bldxi}{{\mbox{\boldmath $\xi$}}}
\newcommand{\bldXi}{{\mbox{\boldmath $\Xi$}}}

\newcommand{\zeros}{{\mbox{\boldmath $0$}}}%

    \def\squarebox#1{\hbox to #1{\hfill\vbox to #1{\vfill}}}

\newlength{\Algwidth}

\begin{document}

\title{Linear-Programming Receivers}

\author{
\authorblockN{Mark F. Flanagan}
\authorblockA{DEIS, University of Bologna \\
via Venezia 52, 47023 Cesena (FC), Italy  \\
Email: mark.flanagan@ieee.org}
}

\maketitle

\begin{abstract}

It is shown that any communication system which admits a sum-product (SP) receiver also admits a corresponding linear-programming (LP) receiver. The two receivers have a relationship defined by the local structure of the underlying graphical model, and are inhibited by the same phenomenon, which we call \emph{pseudoconfigurations}. This concept is a generalization of the concept of \emph{pseudocodewords} for linear codes. It is proved that the LP receiver has the `optimum certificate' property, and that the receiver output is the lowest cost pseudoconfiguration. Equivalence of graph-cover pseudoconfigurations and linear-programming pseudoconfigurations is also proved. 
%There exist pseudoconfigurations which correspond to configurations, and in general, some which do not. 
While the LP receiver is generally more complex than the corresponding SP receiver, the LP receiver and its associated pseudoconfiguration structure provide an analytic tool for the analysis of SP receivers.
As an example application, we show how the LP design technique may be applied to the problem of joint equalization and decoding. 
%and compare performance of the SP receiver with that of the LP receiver. 
%\textbf{Keywords:}
%Linear-programming, sum-product algorithm, factor graphs, pseudoconfigurations, equalization, decoding. 
\end{abstract}

\section{Introduction}

The decoding algorithms for the best known classes of error-correcting code to date, namely concatenated (``turbo") codes \cite{Berrou} and low-density parity check (LDPC) codes \cite{Gallager}, have been shown to be instances of a much more general algorithm called the \emph{sum-product} (SP) algorithm \cite{Wiberg,Aji, Kschischang}. This algorithm solves the general problem of marginalizing a product of functions which take values in a semiring $\cR$. In the communications context, $\cR$ is equal to $\mathbb{R}$ and the maximization of each marginal function minimizes the error rate on a symbol-by-symbol basis. It was also shown that many diverse situations may allow the use of SP based reception \cite{Unified_design}, including joint iterative equalization and decoding (or \emph{turbo equalization}) \cite{turbo_eq} and joint source-channel decoding \cite{Goertz}.

Recently, a linear-programming (LP) based approach to decoding linear (and especially LDPC) codes was developed for binary \cite{Feldman-thesis, Feldman} and nonbinary coding frameworks \cite{FSBG_SCC, FSBG_journal}. The concept of \emph{pseudocodeword} proved important in the performance analysis of both LP and SP based decoders \cite{FKKR, KV-characterization, KV-IEEE-IT}. 
%Also, LP decoders for some classes of turbo codes were described in \cite{Feldman_turbo_IRA}.
Also, linear-programming decoders for irregular repeat-accumulate (IRA) codes and turbo codes were described in \cite{Feldman_turbo_IRA}. Regarding applications beyond coding, an LP-based method for low-complexity joint equalization and decoding of LDPC coded transmissions over the magnetic recording channel was proposed in \cite{Siegel}.

In this paper it is shown that the problem of maximizing a product of $\cR$-valued functions is amenable to an approximate (suboptimal) solution using an LP relaxation, under two conditions: first, that the semiring $\cR$ is equal to a subset of $\mathbb{R}$, and second, that all non-pendant factor nodes are indicator functions for a local behaviour. Fortunately, these conditions are satisfied by almost all practical communication receiver design problems. Interestingly, the LP exhibits a ``separation effect'' in the sense that pendant factor nodes in the factor graph contribute the cost function, and non-pendant nodes determine the LP constraint set. This distinction is somewhat analagous to the case of SP-based reception where pendant factor nodes contribute initial messages exactly once, and all other nodes update their messages periodically. Our LP receiver generalizes the LP \emph{decoders} of \cite{Feldman, FSBG_journal,Feldman_turbo_IRA}. It is proved that both the SP and LP based receivers are inhibited by the same phenomenon which we characterize as a set of \emph{pseudoconfigurations}; this is not intuitively obvious since the SP receiver derives from an attempt to minimize error rate on a symbol-by-symbol basis, while the LP receiver derives from an attempt to minimize the configuration error rate. 

\section{Maximization of a Product of Functions by Linear Programming}

We begin by introducing some definitions and notation. Suppose that we have variables $x_i$, $i \in \cI$, where $\cI$ is a finite set, and the variable $x_i$ lies in the finite set $\cA_i$ for each $i \in \cI$. Let $\bldx = ( x_i )_{i\in\cI}$; then $\bldx$ is called a \emph{configuration}, and the Cartesian product $\cA = \prod_{i \in\cI} \cA_i$ is called the \emph{configuration space}. Suppose now that we wish to find that configuration $\bldx \in \cA$ which maximizes the product of real-valued functions 
\begin{equation}
u\left(\bldx\right)=\prod_{j\in \cJ}f_{j}\left(\bldx_{j}\right)\label{eq:factorization_of_global_function}
\end{equation}
where $\cJ$ is a finite set, $\bldx_j = ( x_i )_{i\in\cI_j}$ and $\cI_j \subseteq \cI$ for each $j \in \cJ$. We define the \emph{optimum} configuration $\bldx_\mathrm{opt}$ to be that configuration $\bldx \in \cA$ which maximizes (\ref{eq:factorization_of_global_function}). The function $u(\bldx)$ is called the \emph{global function} \cite{Kschischang}.

The factor graph for the global function $u(\bldx)$ and its factorization (\ref{eq:factorization_of_global_function}) is a (bipartite) graph defined as follows. There is a variable node for each variable $x_{i}$ ($i \in \cI$) and a factor node for each factor $f_{j}$ ($j \in \cJ$). An edge connects variable node $x_{i}$ to factor node $f_{j}$ if and only if $x_{i}$ is an argument of $f_{j}$. Note that for any $j\in\cJ$, $\cI_j$ is the set of $i\in\cI$ for which $x_{i}$ is an argument of $f_{j}$. Also, for any $i\in\cI$, the set of $j\in\cJ$ for which $x_{i}$ is an argument of $f_{j}$ is denoted $\cJ_i$. The \emph{degree} of a node $v$, denoted $d(v)$, is the number of nodes to which it is joined by an edge. Any node $v$ for which $d(v)=1$ is said to be \emph{pendant}. 

Define
\[
\cY = \{ i \in \cI  \; : \; \exists j \in \cJ_i \: \: \mbox{with} \: \: d(f_j) = 1 \}
\]
i.e., $\cY \subseteq \cI$ is the set of $i \in \cI$ for which variable node $x_i$ is connected to a pendant factor node. Then, for $i\in\cY$, define
\[
h_{i}\left(x_{i}\right) = \prod_{j\in J_i \; : \; d(f_j) = 1} f_j\left(x_{i}\right) \, .
\]
We assume that the function $h_{i}\left(x_{i}\right)$ is positive-valued for each $i\in\cY$. Also, denoting the Cartesian product $\cA_{\cYY} = \prod_{i \in\cY} \cA_i$, we define the projection  
\[
\bldP_{\cYY} \; : \; \cA \longrightarrow \cA_{\cYY} \quad \mbox{such that} \quad \bldP_{\cYY}\left(\bldx\right) = ( x_i )_{i\in\cY} \, .
\]
Also, we adopt the notation $\bldx_{\cYY} = ( x_i )_{i\in\cY}$ for elements of $\cA_{\cYY}$. 

Next define $\cL = \{j\in\cJ \; : \; d(f_{j}) \ge  2 \}$.  
So, without loss of generality we may write 
\begin{equation}
u\left(\bldx\right)=\prod_{i\in \cY}h_{i}\left(x_{i}\right) \cdot \prod_{j\in \cL}f_{j}\left(\bldx_{j}\right)\label{eq:factorization_of_global_function2}
\end{equation}
Now, assume that all factor nodes $f_j$, $j \in\cL$, are indicator functions for some local behaviour $\cB_j$, i.e.,
\[
f_{j}\left(\bldx_{j}\right) = [\bldx_{j} \in \cB_j] \quad \forall j\in \cL 
\]
where the indicator function for the logical predicate $P$ is defined by 
\[
[P] = \left\{ \begin{array}{cc}
1 & \textrm{ if } P \textrm{ is true } \\
0 & \textrm{ otherwise. }\end{array}\right.
\]
Note that we write any $\bldv\in\cB_j$ as $\bldv = ( v_i )_{i\in\cI_j}$, i.e., $\bldv$ is indexed by $\cI_j$. Also we define the \emph{global behaviour} $\cB$ as follows: for any $\bldx \in\cA$, we have $\bldx \in\cB$ if and only if $\bldx_j \in\cB_j$ for every $j\in\cL$. The configuration $\bldx \in \cA$ is said to be \emph{valid} if and only if $\bldx \in \cB$.

We assume that the mapping $\bldP_{\cYY}$ is injective on $\cB$, i.e., if $\bldx_1, \bldx_2 \in \cB$ and $\bldP_{\cYY}(\bldx_1) = \bldP_{\cYY}(\bldx_2)$, then $\bldx_1 = \bldx_2$. This corresponds to a `well-posed' problem. Note that in the communications context, since all non-pendant factor nodes are indicator functions, observations may only be contributed through the set of pendant factor nodes. Therefore, failure of the injectivity property in the communications context would mean that one particular set of channel inputs could correspond to two different transmit information sets, which would reflect badly on system design.

So we have 
\begin{eqnarray*}
\bldx_\mathrm{opt} & = & \arg \max_{\bldxx \in \cA} \left( \prod_{i\in \cY}h_{i}\left(x_{i}\right) \cdot \prod_{j\in \cL}f_{j}\left(\bldx_{j}\right) \right) \\
 & = & \arg \max_{\bldxx \; : \; \bldxx_j \in \cB_j \forall j\in \cL} \prod_{i\in \cY}h_{i}\left(x_{i}\right) \\
 & = & \arg \max_{\bldxx \in \cB} \sum_{i\in \cY} \log h_{i}\left(x_{i}\right) \; .
\end{eqnarray*}
Denote $N_i = |\cA_i|$ for $i \in \cI$, and $N_{\cYY} = \sum_{i \in \cYY} N_i$. Now, for each $i \in \cI$, define the mapping 
\[
\bldxi_i \; : \; \cA_i \longrightarrow \{ 0, 1 \}^{N_i} \subset \mathbb{R}^{N_i}
\]
by
\[
\bldxi_i (\alpha) = ( [\gamma = \alpha] )_{\gamma \in A_i } \; .
\]
Building on these mappings, we also define
\[
\bldXi \; : \; \cA_{\cYY} \longrightarrow \{ 0, 1 \}^{N_{\cYY}} \subset \mathbb{R}^{N_{\cYY}} 
\]
according to
\[
\bldXi(\bldx_{\cYY}) = ( \bldxi_i(x_i) )_{i \in\cY} \; .
\]
We note that $\bldXi$ is injective.

Now, for vectors $\bldg \in \mathbb{R}^{N_{\cYY}}$, we adopt the notation
\[
\bldg = ( \bldg_i )_{i \in \cYY} \quad \mbox{where} \quad \bldg_i = ( g_i^{(\alpha)} )_{\alpha \in \cA_i} \quad \forall i \in \cY \; .
\]
In particular, we define the vector $\boldsymbol{\lambda} \in \mathbb{R}^{N_{\cYY}}$ by setting
\[
\lambda_i^{(\alpha)} = \log h_i(\alpha)
\]
for each $i \in \cY$, $\alpha \in \cA_i$. This allows us to develop the formulation of the optimum configuration as
\begin{eqnarray*}
\bldx_\mathrm{opt} & = & \arg \max_{\bldxx \in \cB} \sum_{i\in \cY} \log h_{i}\left(x_{i}\right) \\
 & = & \arg \max_{\bldxx \in \cB} \sum_{i\in \cY} \boldsymbol{\lambda}_i \bldxi_i (x_i) ^T \\
 & = & \arg \max_{\bldxx \in \cB} \boldsymbol{\lambda} \bldXi (\cP_{\cYY} (\bldx)) ^T \; .
\end{eqnarray*}
Note that the optimization has reduced to the maximization of an inner product of vectors, where the first vector derives only from observations (or ``channel information") and the second vector derives only from the global behaviour (the set of valid configurations). This problem may then be recast as a linear program
\begin{equation}
\bldx_\mathrm{opt} = \bldP_{\cYY}^{-1} \left( \bldXi^{-1} (\bldg_\mathrm{opt}) \right)
\label{eq:xopt_LP1}
\end{equation}
where
\begin{equation}
\bldg_\mathrm{opt} = \arg \max_{\bldgg \in \cK_{\cYY}(\cB)} \boldsymbol{\lambda} \bldg ^T 
\label{eq:gopt_LP1}
\end{equation}
and the maximization is over the convex hull of all points corresponding to valid configurations:
\begin{equation}
\cK_{\cYY}(\cB) = H_\mathrm{conv} \big\{ \bldXi \left( \bldP_{\cYY}\left(\bldx\right) \right) \; : 
\; \bldx\in\cB \big\} \; .
\label{eq:convex_hull_LP1}
\end{equation}

\section{Equivalent Linear Programming Solution for the Optimum Configuration}
\label{sec:equiv_opt_LP}

In this section we define a lower-complexity LP to solve for the optimum configuration, and prove that its performance is equivalent to the original.
Here, for each $i\in\cI$, let $\alpha_i$ be an arbitrary element of $\cA_i$, and let $\cA_i^{-} = \cA_i \backslash \{ \alpha_i \} $ (note that for each $i\in\cI$, $|\cA_i| \ge 2$, otherwise $x_i$ is not a `variable'). Denote $N_i^{-} = |\cA_i^{-}| \ge 1$ for every $i \in \cI$, and denote $N_{\cYY}^{-} = \sum_{i \in \cY} N_i^{-}$. Then, for each $i \in \cI$, define the mapping 
\[
\tilde{\bldxi}_i \; : \; \cA_i \longrightarrow \{ 0, 1 \}^{N_i^{-}} \subset \mathbb{R}^{N_i^{-}}
\]
by 
\[
\tilde{\bldxi}_i (\alpha) = ( [\gamma = \alpha] )_{\gamma \in A_i^{-} } \; .
\]
Building on this, we also define
\[
\tilde{\bldXi} \; : \; \cA_{\cYY} \longrightarrow \{ 0, 1 \}^{N_{\cYY}^{-}} \subset \mathbb{R}^{N_{\cYY}^{-}}
\]
according to
\[
\tilde{\bldXi}(\bldx_{\cYY}) = ( \tilde{\bldxi}_i(x_i) )_{i \in\cY} \; .
\]
We note that $\tilde{\bldXi}$ is injective.

Now, for vectors $\tilde{\bldg} \in \mathbb{R}^{N_{\cYY}^{-}}$, we adopt the notation
\[
\tilde{\bldg} = ( \tilde{\bldg}_i )_{i \in \cYY} \quad \mbox{where} \quad \tilde{\bldg}_i = ( \tilde{g}_i^{(\alpha)} )_{\alpha \in \cA_i^{-}} \quad \forall i \in \cY \; .
\] 
In particular, we define the vector $\tilde{\boldsymbol{\lambda}} \in \mathbb{R}^{N_{\cYY}^{-}}$ by setting
\[
\tilde{\lambda}_i^{(\alpha)} = \log \left[ \frac{ h_i(\alpha) }{ h_i(\alpha_i) } \right]
\]
for each $i \in \cY$, $\alpha \in \cA_i^{-}$. Our new LP is then given by
\begin{equation}
\bldx_\mathrm{opt} = \bldP_{\cYY}^{-1} \left( \tilde{\bldXi}^{-1} (\tilde{\bldg}_\mathrm{opt}) \right)
\label{eq:xopt_LP2}
\end{equation}
where
\begin{equation}
\tilde{\bldg}_\mathrm{opt} = \arg \max_{\tilde{\bldgg} \in \tilde{\cK}_{\cYY}(\cB)} \tilde{\boldsymbol{\lambda}} \tilde{\bldg} ^T 
\label{eq:gopt_LP2}
\end{equation}
and the maximization is over the convex hull of all points corresponding to valid configurations:
\begin{equation} 
\tilde{\cK}_{\cYY}(\cB) = H_\mathrm{conv} \big\{ \tilde{\bldXi} \left( \bldP_{\cYY}\left(\bldx\right) \right) \; : 
\; \bldx\in\cB \big\} \; .
\label{eq:convex_hull_LP2}
\end{equation}

The following proposition proves the equivalence of the original and lower-complexity linear programs.
\begin{proposition}\label{prop:LP_equivalence}
The linear program defined by (\ref{eq:xopt_LP2})--(\ref{eq:convex_hull_LP2}) produces the same (optimum) configuration output as the linear program defined by (\ref{eq:xopt_LP1})--(\ref{eq:convex_hull_LP1}).  
\end{proposition}
\begin{proof}
It is easy to show that the simple bijection 
\[ 
\bldW \; : \; \tilde{\cK}_{\cYY}(\cB) \longrightarrow \cK_{\cYY}(\cB) 
\]
defined by
\[
\bldW(\tilde{\bldg}) = \bldg
\]
where
\[
\forall i \in \cY, \alpha \in \cA_i, \quad g_i^{(\alpha)} = 
\left\{ \begin{array}{cc}
\tilde{g}_i^{(\alpha)} & \textrm{ if } \alpha \in \cA_i^{-} \\
1 - \sum_{\beta \in \cA_i^{-}} \tilde{g}_i^{(\beta)} & \textrm{ if } \alpha = \alpha_i \end{array}\right.
\]
which has inverse given by 
\[
\tilde{g}_i^{(\alpha)} = g_i^{(\alpha)} \quad \forall i \in \cY, \alpha \in \cA_i^{-}
\]
has the property that $\tilde{\bldg} = \tilde{\bldXi}(\cP_{\cYY}(\bldx))$ if and only if $\bldg = \bldW(\tilde{\bldg}) = \bldXi(\cP_{\cYY}(\bldx))$. Also observe that if $\bldg = \bldW(\tilde{\bldg})$,
\begin{eqnarray}
\tilde{\boldsymbol{\lambda}} \tilde{\bldg} ^T & = & \sum_{\alpha \in \cA_i^{-}} \left[ \log h_i(\alpha) - \log h_i(\alpha_i) \right] \tilde{g}_i^{(\alpha)} \nonumber \\
& = & \sum_{\alpha \in \cA_i^{-}} \log h_i(\alpha) g_i^{(\alpha)} - \log h_i(\alpha_i) [ 1 - g_i^{(\alpha_i)} ] \nonumber \\
& = & \boldsymbol{\lambda} \bldg ^T- \log h_i(\alpha_i) 
\label{eq:cost_fn_equivalence}
\end{eqnarray}
i.e. the bijection $\bldW$ preserves the cost function up to an additive constant. Taken together, these two facts imply that $\tilde{\bldg}$ maximizes $\tilde{\boldsymbol{\lambda}} \tilde{\bldg} ^T$ over $\tilde{\cK}_{\cYY}(\cB)$ if and only if  $\bldg = \bldW(\tilde{\bldg})$ maximizes $\boldsymbol{\lambda} \bldg ^T$ over $\cK_{\cYY}(\cB)$. This proves the result, and justifies the use of the notation $\bldx_\mathrm{opt}$ in (\ref{eq:xopt_LP2}).
\end{proof}

%-------------------------------------------------------------------------------------- 
\section{Efficient Linear Programming Relaxation and its Properties}
\label{sec:efficient_LP}

To reduce complexity of the LP, we introduce auxiliary variables whose constraints, along with those of the elements of $\tilde{\bldg} \in \mathbb{R}^{N_{\cYY}^{-}}$,
will form the relaxed LP problem. We denote these auxiliary variables by
\begin{equation}
p_{j,\bldbb} \; \mbox{ for each } \;  j\in\cL, \bldb \in \cB_j
\label{eq:auxiliary_variables}
\end{equation}
and we form the following vector
\[
\bldp = \left( \bldp_j \right)_{j \in \cL} \; \mbox{ where } \; \bldp_j = \left( p_{j,\bldbb} \right)_{\bldbb \in \cB_j} \: \forall j \in \cL \; .
\]
For each $i\in\cI\backslash \cY$, let $t_i$ be an arbitrary element of $J_i$. The constraints of the relaxed LP problem are then
\begin{equation}
\forall j \in \cL, \; \forall \bldb \in \cB_j,  \quad  p_{j,\bldbb} \ge 0 \; ,
\label{eq:equation-polytope-4} 
\end{equation} 
\begin{equation}
\forall j \in \cL, \quad \sum_{\bldbb \in \cB_j} p_{j,\bldbb} = 1 \; ,
\label{eq:equation-polytope-5} 
\end{equation}  
\begin{eqnarray}
& \forall i \in \cY, \; \forall j \in \cJ_i \cap \cL, \; \forall \alpha \in \cA_i^{-}, \nonumber \\
& \tilde{g}_i^{(\alpha)} = 
\sum_{\bldbb \in \cB_j, \; b_i=\alpha} p_{j,\bldbb}
\label{eq:equation-polytope-6} 
\end{eqnarray} 
and
\begin{eqnarray}
& \forall i \in \cI\backslash\cY, \; \forall j \in \cJ_i \backslash \{t_i\}, \; \forall \alpha \in \cA_i^{-}, \nonumber \\
& \sum_{\bldbb \in \cB_j, \; b_i=\alpha} p_{j,\bldbb} = 
\sum_{\bldbb \in \cB_{t_i}, \; b_i=\alpha} p_{t_i,\bldbb} \; .
\label{eq:equation-polytope-7} 
\end{eqnarray} 

Constraints~(\ref{eq:equation-polytope-4})-(\ref{eq:equation-polytope-7}) form a polytope which we denote by $\tilde{\cQ}$. The maximization of
the objective function $\tilde{\boldsymbol{\lambda}} \tilde{\bldg} ^T$ over $\tilde{\cQ}$ forms the relaxed LP problem.

Observe that the further constraints
\begin{equation}
\forall j \in \cL, \; \forall \bldb \in \cB_j,  \quad  p_{j,\bldbb} \le 1 \; ,
\label{eq:equation-polytope-1} 
\end{equation}  
\begin{equation}
\forall i \in \cY, \; \forall \alpha\in\cA_i^{-}, \quad 0 \le \tilde{g}_i^{(\alpha)} \le 1   
\label{eq:equation-polytope-2} 
\end{equation} 
and
\begin{equation}
\forall i \in \cY, \quad \sum_{\alpha\in\cA_i^{-}} \tilde{g}_i^{(\alpha)} \le 1 
\label{eq:equation-polytope-3} 
\end{equation} 
follow from the constraints~(\ref{eq:equation-polytope-4})-(\ref{eq:equation-polytope-7}), for any 
$(\tilde{\bldg}, \bldp) \in \tilde{\cQ}$.

The receiver algorithm works as follows. First, we say a point in a polytope is \emph{integral} if and only if all of its coordinates are integers. If the LP solution $(\tilde{\bldg}_\mathrm{out}, \bldp)$ is an integral point in $\tilde{\cQ}$, the output is the configuration $\bldx_\mathrm{out} = \bldP_{\cYY}^{-1} \left( \tilde{\bldXi}^{-1} (\tilde{\bldg}_\mathrm{out}) \right)$ (we shall prove in the next section that this output is indeed in $\cB$). This configuration may be equal to the optimum configuration (we call this `correct reception') or it may not be (we call this `incorrect reception'). Of course, in the communications context, we are usually only interested in a subset of the configuration symbols, namely the information bits. If the LP solution is not integral, the receiver reports a `receiver failure'. Note that in this paper, we say that the receiver makes a \emph{reception error} when the receiver output is not equal to the correct configuration (this could correspond to a `receiver failure', or to an `incorrect reception').

\section{LP Equivalent to the Efficient Relaxation}
\label{sec:efficient_LP_theoretical}

We next define another linear program, and prove that its performance is equivalent to that defined in Section \ref{sec:efficient_LP}. This new program is more computationally complex than that defined previously, but (due to its equivalence and simplicity of description) is more useful for theoretical work. For this we define $N = \sum_{i \in \cI} N_i$ and 
\[
\bar{\bldXi} \; : \; \cA \longrightarrow \{ 0, 1 \}^{N} \subset \mathbb{R}^{N}  
\]
according to
\[
\bar{\bldXi}(\bldx) = ( \bldxi_i(x_i) )_{i \in\cI} \; .
\]
Again, $\bar{\bldXi}$ is injective. For vectors $\bar{\bldg} \in \mathbb{R}^{N}$, we denote $\bar{\bldg} = ( \bldg_i )_{i \in \cI}$ and $\bldg = ( \bldg_i )_{i \in \cY}$, where $\bldg_i = ( g_i^{(\alpha)} )_{\alpha \in \cA_i}$ for each $i \in \cI$.

The new LP optimizes the cost function $\boldsymbol{\lambda} \bldg ^T$ over the polytope $\cQ$ defined with respect to variables $\bar{\bldg}$ and $\bldp$, the constraints (\ref{eq:equation-polytope-4}) and (\ref{eq:equation-polytope-5}), and the following single constraint: 
\begin{equation}
\forall j \in \cL, \; \forall i \in \cI_j, \; \forall \alpha \in \cA_i, \quad g_i^{(\alpha)} = \sum_{\bldbb \in \cB_j, \; b_i=\alpha} p_{j,\bldbb} \; .
\label{eq:equation-polytope-theoretical} 
\end{equation} 
Note that (\ref{eq:equation-polytope-1}), together with the constraints
\begin{equation}
\forall i \in \cI, \; \forall \alpha\in\cA_i, \quad 0 \le g_i^{(\alpha)} \le 1   
\label{eq:equation-polytope-2-new} 
\end{equation} 
and
\begin{equation}
\forall i \in \cI, \quad \sum_{\alpha\in\cA_i} g_i^{(\alpha)} = 1 
\label{eq:sum_g_equals_1} 
\end{equation} 
follow from the constraints~(\ref{eq:equation-polytope-4}), (\ref{eq:equation-polytope-5}) and (\ref{eq:equation-polytope-theoretical}), for any 
$(\bar{\bldg}, \bldp) \in \cQ$.

In this case, the receiver output is equal to the configuration $\bldx_\mathrm{out} = \bar{\bldXi}^{-1} (\bar{\bldg}_\mathrm{out})$ in the case where the LP solution $(\bar{\bldg}_\mathrm{out}, \bldp)$ is an integral point in $\cQ$ (again, this output is in $\cB$), and reports a `receiver failure' if the LP solution is not integral. 

The following theorem ensures the equivalence of the two linear programs, and also assures the \emph{optimum certificate} property, i.e., if the receiver output is a configuration, then it is the optimum configuration.
\begin{theorem}\label{prop:LP_equivalence_2}
The linear program defined by (\ref{eq:equation-polytope-4}), (\ref{eq:equation-polytope-5}) and (\ref{eq:equation-polytope-theoretical}) produces the same output (configuration or receiver failure) as the linear program defined by (\ref{eq:equation-polytope-4})--(\ref{eq:equation-polytope-7}). Also, in the case of configuration output, if the receiver output is a configuration, then it is the optimum configuration, i.e., $\bldx_\mathrm{out} = \bldx_\mathrm{opt}$. 
\end{theorem}
\begin{proof}
It is straightforward to show that the mapping
\begin{eqnarray*}
\bldV \; : \; \tilde{\cQ} & \longrightarrow & \cQ \\
(\tilde{\bldg},\bldp) & \mapsto  & (\bar{\bldg},\bldp)
\end{eqnarray*}
defined by
\[
g_i^{(\alpha)} = 
\left\{ \begin{array}{ccc}
\tilde{g}_i^{(\alpha)} & \textrm{ if } i \in \cY, \alpha \in \cA_i^{-} \\
1 - \sum_{\beta \in \cA_i^{-}} \tilde{g}_i^{(\beta)} & \textrm{ if }  i \in \cY, \alpha = \alpha_i \\
\sum_{\bldbb \in \cB_{t_i}, \; b_i=\alpha} p_{{t_i},\bldbb}  & \textrm{ if } i \in \cI \backslash \cY \end{array}\right.
\]
and with inverse 
\[
\tilde{g}_i^{(\alpha)} = g_i^{(\alpha)} \quad \forall i \in \cY, \alpha \in \cA_i^{-}
\]
is a bijection from one polytope to the other (i.e. $\tilde{\bldg} \in \mathbb{R}^{N_{\cYY}^{-}}$ satisfies (\ref{eq:equation-polytope-4})--(\ref{eq:equation-polytope-7}) for some vector $\bldp$ if and only if $\bar{\bldg} \in \mathbb{R}^{N}$ with $(\bar{\bldg}, \bldp) = \bldV(\tilde{\bldg}, \bldp)$ satisfies (\ref{eq:equation-polytope-4}), (\ref{eq:equation-polytope-5}) and (\ref{eq:equation-polytope-theoretical}) for the same vector $\bldp$). Also, since $(\bar{\bldg}, \bldp) = \bldV(\tilde{\bldg}, \bldp)$ implies $\bldg = \bldW(\tilde{\bldg})$, (\ref{eq:cost_fn_equivalence}) implies that the bijection $\bldV$ preserves the cost function up to an additive constant. 

Next, we prove that for every configuration $\bldx \in \cB$, there exists $\bldp$ such that $(\bar{\bldXi} \left(\bldx \right), \bldp) \in \cQ$. Let $\bldx \in \cB$, and define
\[
\forall j \in \cL, \bldb \in \cB_j, \quad p_{j,\bldbb} = 
\left\{ \begin{array}{cc}
1 & \textrm{ if } \bldb = (x_i)_{i \in \cI_j} \\
0 & \textrm{ otherwise. } \end{array}\right.
\]
Letting $\tilde{\bldg} = \tilde{\bldXi}(\cP_{\cYY}(\bldx))$ and $\bar{\bldg} = \bar{\bldXi}(\bldx)$, it is easy to check that $(\tilde{\bldg}, \bldp) \in \tilde{\cQ}$ and $(\bar{\bldg}, \bldp) \in \cQ$ (and that in fact $(\bar{\bldg},\bldp) = \bldV(\tilde{\bldg}, \bldp)$). This property ensures that every valid configuration $\bldx \in \cB$ has a ``representative" in the polytope, and thus is a candidate for being output by the receiver.

Next, let $(\tilde{\bldg}, \bldp) \in \tilde{\cQ}$ and let $\bar{\bldg} \in \mathbb{R}^{N}$ be such that $(\bar{\bldg}, \bldp) = \bldV(\tilde{\bldg}, \bldp) \in \cQ$. Suppose that all of the coordinates of $\bldp$ are integers. Then, by (\ref{eq:equation-polytope-4}) and (\ref{eq:equation-polytope-5}), for any $j \in \cL$ we must have 
\[
\forall \bldb \in \cB_j, \quad p_{j,\bldbb} = 
\left\{ \begin{array}{cc}
1 & \textrm{ if } \bldb = \bldb^{(j)} \\
0 & \textrm{ otherwise } \end{array}\right.
\]
for some $\bldb^{(j)} \in \cB_j$.

Now we note that for any $i \in \cI$, $j,k \in \cJ_i \cap \cL$, if $b^{(j)}_i = \alpha$ then (using (\ref{eq:equation-polytope-theoretical}))
\begin{equation}
g_i^{(\alpha)} = \sum_{\bldbb \in \cB_j, \; b_i=\alpha} p_{j,\bldbb} = 1 = \sum_{\bldbb \in \cB_k, \; b_i=\alpha} p_{k,\bldbb}
\label{eq:starstar}
\end{equation}
and thus $b^{(j)}_i = \alpha$. Therefore, there exists $\bldx \in \cA$ such that 
\[
(x_i)_{i \in \cI_j} = \bldb^{(j)} \quad \forall j \in \cL \; .
\]
Therefore, $\bldx$ is a valid configuration ($\bldx \in \cB$). Also we may conclude from (\ref{eq:starstar}) that
\[
g_i^{(\alpha)} = 
\left\{ \begin{array}{cc}
1 & \textrm{ if } x_i = \alpha \\
0 & \textrm{ otherwise } \end{array}\right.
\] 
and therefore $\bar{\bldg} = \bar{\bldXi} \left( \bldx \right)$. Also, from the definition of the mapping $\bldV$, we have $\tilde{\bldg} = \tilde{\bldXi} \left( \bldP_{\cYY}\left(\bldx\right) \right)$.

Summarizing these results, we conclude that $(\bar{\bldg}_\mathrm{opt}, \bldp) \in \cQ$ optimizes the cost function $\boldsymbol{\lambda} \bldg ^T$ over $\cQ$ and is integral if and only if $(\tilde{\bldg}_\mathrm{opt}, \bldp) = \bldV^{-1}(\bar{\bldg}_\mathrm{opt}, \bldp) \in \tilde{\cQ}$ optimizes the cost function $\tilde{\boldsymbol{\lambda}} \tilde{\bldg} ^T$ over $\tilde{\cQ}$ and is integral, where $\bar{\bldXi}^{-1}(\bar{\bldg}) = \bldx \in \cB$ and $\bldx_{\cYY} = \tilde{\bldXi}^{-1}(\tilde{\bldg}) = \cP_{\cYY}(\bldx)$.
\end{proof}
\medskip
Thus both LP receivers output either a receiver failure, or the optimum configuration, and have the same performance. The LP receiver of Section \ref{sec:efficient_LP} has lower complexity and is suitable for implementation (e.g. for the program of Section \ref{sec:equalization_decoding}); however, for theoretical work the LP of Section \ref{sec:efficient_LP_theoretical} is more suitable (we shall use this polytope throughout Section \ref{sec:PCFs}). 
%-------------------------------------------------------------------------------------- 
\section{Pseudoconfigurations} 
\label{sec:PCFs}

In this section, we prove a connection between the failure of the LP and SP receivers based on \emph{pseudoconfiguration} concepts.
 
\subsection{Linear Programming Pseudoconfigurations}  
\begin{definition}
A \emph{linear-programming pseudoconfiguration} (LP pseudoconfiguration) is a point $(\bar{\bldg}, \bldp)$ in the polytope $\cQ$ with rational coordinates.
\end{definition}

Note that, since the coefficients of the LP are rational, the LP output must be the LP pseudoconfiguration which minimizes the cost function.

\subsection{Factor Graph Covers and Graph-Cover Pseudoconfigurations}

We next define what is meant by a finite cover of a factor graph.  
\begin{definition}
Let $M$ be a positive integer, and let $\cS = \{ 1, 2, \cdots, M \}$. Let $\cG$ be the factor graph corresponding to the global function $u$ and its factorization given in (\ref{eq:factorization_of_global_function}). A \emph{cover configuration} of degree $M$ is a vector 
$\bldx^{(M)} = ( \bldx_i^{(M)} )_{i \in \cI}$ where $\bldx_i^{(M)} = (x_{i,l})_{l \in S} \in \cA_i^M$ for each $i \in \cI$.
Define $u^{(M)}$ as the following function of the cover configuration $\bldx^{(M)}$ of degree $M$:   
\begin{equation}
u^{(M)}\left(\bldx^{(M)}\right)=\prod_{l\in \cS} \prod_{j\in \cJ}f_{j}\left(\bldx_{j,l}\right)
\label{eq:cover_graph_factorization}
\end{equation}
where, for each $j\in\cJ$, $i\in\cI_j$, $\Pi_{j,i}$ is a permutation on the set $S$, and for each $j\in\cJ$, $l\in\cS$,
\[
\bldx_{j,l} = ( x_{i,\Pi_{j,i}(l)} )_{i \in\cI_j} \; .
\]
A \emph{cover} of the factor graph $\cG$, of degree $M$, is a factor graph for the global function $u^{(M)}$ and its factorisation (\ref{eq:cover_graph_factorization}). In order to distinguish between different factor node labels, we write (\ref{eq:cover_graph_factorization}) as
\[
u^{(M)}\left(\bldx^{(M)}\right)=\prod_{l\in \cS} \prod_{j\in \cJ}f_{j,l}\left(\bldx_{j,l}\right)
\]
where $f_{j,l} = f_{j}$ for each $j\in\cJ$, $l\in\cS$.
\end{definition}
\medskip
It may be seen that a cover graph of degree $M$ is a graph whose vertex set consists of $M$ copies of $x_i$ (labelled $x_{i,l}$) and $M$ copies of $f_j$ (labelled $f_{j,l}$), such that for each $j\in\cJ$, $i\in\cI_j$, the $M$ copies of $x_i$ and the $M$ copies of $f_j$ are connected in a one-to-one fashion determined by the permutations $\{ \Pi_{j,i} \}$.

We define the \emph{cover behaviour} $\cB_M$ as follows. The cover configuration $\bldx^{(M)}$ lies in $\cB_M$ if and only if $\bldx_{j,l} \in \cB_j$ for each $j\in\cJ$, $l\in\cS$.

For any $M\ge 1$, a \emph{graph-cover pseudoconfiguration} is a valid cover configuration (i.e. one which lies in the behaviour $\cB_M$). 

For any graph-cover pseudoconfiguration, we also define the \emph{graph-cover pseudoconfiguration vector} $\bar{\bldh} \in \mathbb{R}^{N}$ according to 
\[
\bar{\bldh} = ( \bldh_i )_{i \in \cI} \quad \mbox{where} \quad \bldh_i = ( h_i^{(\alpha)} )_{\alpha \in \cA_i} \quad \forall i\in\cI
\]
and
\[
h_i^{(\alpha)} = \left| \{ l \in \cS \; : \; x_{i,l} = \alpha \} \right| 
\]
for each $i\in\cI$, $\alpha\in\cA_i$. Finally, we define the \emph{normalized graph-cover pseudoconfiguration vector} $\bar{\bldg} \in \mathbb{R}^{N}$ by $\bar{\bldg} = \bar{\bldh} / M$.

\subsection{Equivalence between Pseudoconfiguration Concepts}

In this section, we show the equivalence between the set of LP pseudoconfigurations and the set of
graph-cover pseudoconfigurations. The result is summarized in the following theorem. 
\begin{theorem}\label{thm:PCW_equivalence}
There exists an LP pseudoconfiguration $(\bar{\bldg}, \bldp)$ if and only if there 
exists a graph-cover pseudoconfiguration with normalized pseudoconfiguration vector $\bar{\bldg}$. 
\end{theorem}
\begin{proof}
Suppose $\bldx^{(M)}$ is a graph-cover pseudoconfiguration for some cover of degree $M$ of the factor graph, and $\bar{\bldg}$ is its normalized graph-cover pseudoconfiguration vector. Then, 
\[
g_i^{(\alpha)} = \frac{1}{M} \left| \{ l \in \cS \; : \; x_{i,l} = \alpha \} \right| 
\]
for all $i \in \cI$, $\alpha \in \cA_i$. Next define $\bldp$ according to 
\[
p_{j,\bldbb} = \frac{1}{M} \left| \{ l \in \cS \; : \; \bldx_{j,l} = \bldb \} \right| 
\]
for all $j \in \cL$, $\bldb \in \cB_j$. Then it is easily seen that (\ref{eq:equation-polytope-4}), (\ref{eq:equation-polytope-5}) and (\ref{eq:equation-polytope-theoretical}) are satisfied, and so $(\bar{\bldg},\bldp) \in \cQ$.

To prove the other direction, suppose $(\bar{\bldg},\bldp) \in \cQ$. Denote by $M$ the lowest common denominator of the (rational) variables $p_{j,\bldbb}$ for $j\in\cL$, $\bldb \in\cB_j$. Define $z_{j,\bldbb} = M p_{j,\bldbb}$ for $j\in\cL$, $\bldb \in\cB_j$; these must all be nonnegative integers. Also define $h_i^{(\alpha)} = M g_i^{(\alpha)}$ for all $i \in \cI$, $\alpha \in \cA_i$; these must all be nonnegative integers by (\ref{eq:equation-polytope-theoretical}).   

We now construct a cover graph of degree $M$ as follows. Begin with $M$ copies of vertex $x_i$ (labelled $x_{i,l}$) and $M$ copies of vertex $f_j$ (labelled $f_{j,l}$), for $i \in \cI$, $j \in \cJ$. Then proceed as follows:

\begin{itemize}
\item
Label $h_i^{(\alpha)}$ copies of $x_i$ with the value $\alpha$, for each $i \in \cI$, $\alpha \in \cA_i$. By (\ref{eq:sum_g_equals_1}), all copies of $x_i$ are labelled. 

\item
Label $z_{j,\bldbb}$ copies of $f_j$ with the value $\bldb$, for every $j \in \cL$, $\bldb \in \cB_j$. By (\ref{eq:equation-polytope-5}), all copies of $f_j$ are labelled.

\item
Next, let $T_i^{(\alpha)}$ denote the set of copies of $x_i$ labelled with the value $\alpha$, for $i\in\cI$, $\alpha\in\cA_i$. Also, for all $j\in\cL$, $i\in\cI_j$, $\alpha\in\cA_i$, let $R_{i,j}^{(\alpha)}$ denote the set of copies of $f_j$ whose label $\bldb$ satisfies $b_i=\alpha$. The vertices in $T_i^{(\alpha)}$ and the vertices in $R_{i,j}^{(\alpha)}$ are then connected by edges in an arbitrary one-to-one fashion, for every $j\in\cL$, $i\in\cI_j$, $\alpha\in\cA_i$.

Numerically, this connection is always possible because
\begin{equation*}
| T_i^{(\alpha)} | = h_i^{(\alpha)} = \sum_{\bldbb \in \code_j, \; b_i=\alpha} z_{j,\bldbb} = | R_{i,j}^{(\alpha)} | 
\end{equation*}
for every $j\in\cL$, $i\in\cI_j$, $\alpha \in \cA_i$. Here we have used~(\ref{eq:equation-polytope-theoretical}). Finally, for each $i \in \cY$ and $j \in \cJ_i$ satisfying $d(f_j) = 1$, the $M$ copies of $x_i$ are connected to the $M$ copies of $f_j$ in an arbitrary one-to-one fashion.
It is easy to check that the resulting graph is a cover graph for the original factorization. Therefore, this vertex labelling yields a graph-cover pseudoconfiguration.  

\end{itemize}
\end{proof}

%-----------------------------------------------------------------------
\section{Example Application: LP-Based Joint Equalization and Decoding} 
\label{sec:equalization_decoding}

In this section we consider an example application where we use the above framework to design an LP receiver for a system using binary coding and binary phase-shift keying (BPSK) modulation over a frequency selective channel (with memory $L$) with additive white Gaussian noise (AWGN). Information-bearing data are encoded to form codewords of the binary code 
\[
\code = \{ \bldc \in GF(2)^n \; : \; \bldc \cH^T = 0 \}  
\]
where $\cH$ is the code's $m \times n$ \emph{parity-check matrix} over $GF(2)$. Denote the set of code bit indices and parity-check indices by $\cU = \{ 1,2,\cdots, n\}$ and $\cV = \{ 1,2,\cdots, m\}$ respectively. We factor the indicator function for the code into factors corresponding to each local parity check: for $j \in \cV$, define the single-parity-check code over $GF(2)$ by 
\[
\code_j = \{ (c_i)_{i \in \cU_j} \; : \; \sum_{i \in \cU_j} c_i = 0 \}
\]
where $\cU_j \subseteq \cU$ is the support of the $j$-th row of $\cH$ for each $j \in \cV$, and the summation is over $GF(2)$. Thus $\bldc \in \code$ if and only if $(c_i)_{i \in \cU_j} \in \code_j$ for each $j \in \cV$. The BPSK modulation mapping $\cM$, which maps from $GF(2)$ to $\mathbb{R}$, is given by $\cM(0)=1$ and $\cM(1)=-1$. The channel has $L+1$ taps $\{ h_0, h_1, \cdots, h_L\}$ and the received signal for the communication model may be written as
\[
r_i = \sum_{t=0}^{L} h_t \cM(c_{i-t}) + n_i
\]
where $n_i$ is a zero-mean complex Gaussian random variable with variance $\sigma^2$. 
%A diagram of the transmitter-channel model is given in figure X. This represents an equivalent channel model; in reality, the modulator comes before the channel delay line. 

We adopt a state-space (trellis) representation for the channel, with state space $\cS = GF(2)^L$; also let $\cS^{-} = \cS \backslash \{ \zeros \}$. The local behaviour (or trellis edge set) for the state-space model, denoted $\cD$, is assumed to be time-invariant, although extension to the case of time-variant channel is straightforward. For simplicity we assume that the initial and final states of the channel are not known at the receiver. For $\bldd \in \cD$, let $\ip(\bldd)$, $\op(\bldd)$, $s^{S}(\bldd)$ and $s^{E}(\bldd)$ denote the channel input, output, start state and end state respectively. Thus if we set $\cD = GF(2)^{L+1}$ and adopt the notation $\bldd = (d_0 \; d_1 \; \cdots \; d_L) \in \cD$, we may have $\ip(\bldd) = d_0$, $s^S(\bldd) = (d_1 \; d_2 \; \cdots \; d_L)$, $s^E(\bldd) = (d_0 \; d_1 \; \cdots \; d_{L-1})$, and $\op(\bldd) = \sum_{t=0}^{L} h_t \cM(d_t)$. 
Also let $\cD^{-} = \cD \backslash \{ \zeros \}$.

%The code is of length $n=10$, and consist of $m=5$ parity checks. Here we write $\cI = \{ 1, 2, \cdots, n \}$ and $\cJ = \{ 1, 2, \cdots, m \}$. The factor graph is shown in figure ; indicated on the figure are variables and their corresponding LP variables.

%We provide the LP listing in general case, but for simulation results we choose a short binary code with many four-cycles. While impractical, this scheme's performance is poor enough to allow us to compare performance in the waterfall and error floor regions with SP decoder, and to investigate the effect of the pseudoconfiguration phenomenon on system performance in both LP and SP decoding.  

%The structure of the LP may be read off the factor graph, and is as follows. The variables of the LP are (indicator functions for the code bits, channel outputs, parity check words and channel states respectively)
The structure of the LP may be written as follows (we use the efficient LP defined in Section \ref{sec:efficient_LP}). The variables of the LP are \footnote{On a correct reception, the LP variables $\{ w_{j,\bldbb} \}$ and $\{ q_{i,\blddd} \}$ serve as indicator functions for the single-parity-check codewords and the channel outputs respectively.}
\[
\tilde{g}_{i}^{(\blddd)} \; \mbox{ for each } \;  i\in\cU, \bldd \in \cD^{-} \; ,
\]
\[
w_{j,\bldbb} \; \mbox{ for each } \;  j\in\cV, \bldb \in \code_j \; ,
\]
\[
q_{i,\blddd} \; \mbox{ for each } \;  i\in\cU, \bldd \in \cD \; ,
\]
and the LP constraints are (here $\cU^{-} = \cU\backslash\{n\}$)
\begin{equation}
\forall j \in \cV, \; \forall \bldb \in \code_j,  \quad  w_{j,\bldbb} \ge 0 
\end{equation} 
and
\begin{equation}
\forall i \in \cU, \; \forall \bldd \in \cD,  \quad  q_{i,\blddd} \ge 0
\end{equation} 
which follow from (\ref{eq:equation-polytope-4}), 
\begin{equation}
\forall j \in \cV, \quad \sum_{\bldbb \in \code_j} w_{j,\bldbb} = 1
\end{equation}
and
\begin{equation}
\forall i \in \cU, \quad \sum_{\blddd \in \cD} q_{i,\blddd} = 1
\end{equation}
from (\ref{eq:equation-polytope-5}), 
\begin{equation}
\forall i \in \cU, \bldd \in \cD^{-}, \quad \tilde{g}_{i}^{(\blddd)} =  q_{i,\blddd}
\label{eq:sumg_equals_sumq}
\end{equation}
from (\ref{eq:equation-polytope-6}), and
\begin{eqnarray}
& \forall i \in \cU, \forall j \in \cU_j, \nonumber \\
& \sum_{\blddd \in \cD, \; \ip(\blddd)=0} q_{i,\blddd} = \sum_{\bldbb \in \code_j, \; \bldbb_i=0} w_{j,\bldbb}
\end{eqnarray}
and
\begin{eqnarray}
& \forall i \in \cU^{-}, \forall s \in \cS^{-}, \nonumber \\
& \sum_{\blddd \in \cD, \; s^{E}(\blddd)=s} q_{i,\blddd} = \sum_{\blddd \in \cD, \; s^{S}(\blddd)=s} q_{i+1,\blddd}
\end{eqnarray}
from (\ref{eq:equation-polytope-7}).

Also, the LP cost function is
\begin{equation}
\sum_{i \in \cU} \sum_{\blddd \in\cD^{-}} \tilde{\lambda}_i^{(\blddd)} \tilde{g}_i^{(\blddd)}
\label{eq:cost_fn_turbo_eq}
\end{equation}
where we have, for $i\in\cU$, $\bldd \in\cD^{-}$, 
%\begin{multline}
%\tilde{\lambda}_i^{(\blddd)} = \log h_i(\bldd) = -\frac{1}{\sigma^2} \Big( \left| r_i - \sum_{t=0}^{L} h_t \cM(d_t) \right|^2 - \\
%\left| r_i - \sum_{t=0}^{L} h_t \right|^2 \Big)
%\end{multline} 
\begin{eqnarray}
\tilde{\lambda}_i^{(\blddd)} & = & \log \left( \frac{p(r_i | \bldd)}{p(r_i | \zeros)} \right) \nonumber \\ 
& = & \frac{1}{\sigma^2} \left( \left| r_i - \op(\zeros) \right|^2 - \left| r_i - \op(\bldd) \right|^2 \right) \; .
\end{eqnarray}
%Simulation results, in terms of bit error rate (BER) and frame error rate (FER), for the system are shown in figure Z.

Note that in this application, the variables $\tilde{g}_i^{(\blddd)}$, together with the constraint (\ref{eq:sumg_equals_sumq}), would be removed due to their redundancy, and the variables $q_{i,\blddd}$ used directly in the cost function (\ref{eq:cost_fn_turbo_eq}). The resulting LP is capable of joint equalization and decoding, and has strong links (via Theorems \ref{prop:LP_equivalence_2} and \ref{thm:PCW_equivalence}) to the corresponding ``turbo equalizer'' based on application of the sum-product algorithm to the same factorization of the global function.

\section*{Acknowledgment}
The author would like to acknowledge the support of the Institute of Advanced Studies, University of Bologna (ISA-ESRF Fellowship).

%-----------------------------------------------------------------------

\end{document}